\documentclass[sigplan,10pt]{acmart}

\usepackage{xcolor}
\usepackage{listings}
\usepackage{amsthm}
\usepackage{subcaption} 
\usepackage[colorinlistoftodos]{todonotes}

\theoremstyle{definition}

\newtheorem{theorem}{Theorem}[section]

\lstdefinestyle{mystyle}{
    mathescape,
    escapeinside=||,
    frame=lines,
    numbersep=4pt,
    language=C++,
    commentstyle=\color{codegreen},
    keywordstyle=\bfseries\color{blue},
    stringstyle=\color{codepurple},
    basicstyle=\ttfamily\scriptsize,
    breakatwhitespace=false,         
    breaklines=true,                 
    captionpos=b,                    
    keepspaces=true,                 
    numbers=left,                    
    numbersep=5pt,                  
    showspaces=false,                
    showstringspaces=false,
    showtabs=false,                  
    tabsize=2
}

\begin{document}

\title{GPA: A GPU Performance Advisor Based on Instruction Sampling}

\author{Keren Zhou}
\email{keren.zhou@rice.edu}
\affiliation{%
  \institution{Rice University}
  \city{Houston}
  \state{Texas}
  \country{United States}}

\author{Xiaozhu Meng}
\email{xm13@rice.edu}
\affiliation{%
  \institution{Rice University}
  \city{Houston}
  \state{Texas}
  \country{United States}}

\author{Ryuichi Sai}
\email{ryuichi@rice.edu}
\affiliation{%
  \institution{Rice University}
  \city{Houston}
  \state{Texas}
  \country{United States}}

\author{John Mellor-Crummey}
\email{johnmc@rice.edu}
\affiliation{%
  \institution{Rice University}
  \city{Houston}
  \state{Texas}
  \country{United States}}
\begin{abstract}
Developing efficient GPU kernels can be difficult because of the complexity of 
GPU architectures and programming models.
Existing performance tools only provide coarse-grained suggestions at the kernel level, if any.
In this paper, we describe GPA, a performance advisor for NVIDIA GPUs that suggests potential code optimization opportunities at a hierarchy of levels, including individual lines, loops, and functions.
To relieve users of the burden of interpreting performance counters and analyzing bottlenecks, GPA uses data flow analysis to approximately attribute measured instruction stalls to their root causes and uses information about a program's structure and the GPU to match inefficiency patterns with suggestions for optimization.
To quantify each suggestion's potential benefits, we developed PC sampling-based performance models to estimate its speedup.
Our experiments with benchmarks and applications show that GPA provides an insightful report to guide performance optimization.
Using GPA, we obtained speedups on a Volta V100 GPU ranging from 1.01$\times$ to 3.53$\times$, with a geometric mean of 1.22$\times$.
\end{abstract}

\maketitle

\section{Introduction}

Graphics Processing Units (GPUs) have been extensively employed in data centers and supercomputers as a building block to accelerate High-Performance Computing (HPC) and machine learning applications.
However, fully utilizing the compute power of GPUs is challenging.
Tuning GPU code to achieve the maximum possible performance requires significant manual effort to cope with the complexity of GPU architectural 
features and programming models.

GPU profilers~\cite{nvprof, nsightsystem, nsightcompute, hpctoolkit, shende2006tau, scorep, january2015allinea} are widely used for measuring GPU-accelerated applications.
While these tools identify hot GPU code, they lack sophisticated analysis of performance bottlenecks and provide little insight into how to improve the code.
nvprof and Nsight-Compute, for example, analyze performance measurement data and propose suggestions on the kernel level but do not identify specific 
lines that could be optimized nor estimate the potential gain after applying optimizations.
As a result, even with GPU profilers, diagnosing and fixing performance problems requires expertise in interpreting measurement data and associating suggestions with corresponding bottlenecks.

Prior tools on GPUs~\cite{shen2018cudaadvisor, gvprof, braun2019cuda} provide fine-grained suggestions using instrumentation-based methods to quantify the severity of performance problems and locate problematic code.
These tools identify one or a few patterns, such as redundant value/address, insufficient cache utilization, or memory transaction burst, but overlook others.
Moreover, they do not correlate execution time with the patterns.
As a result, one may fix specific problems indicated by the tools but not achieve any speedup.

Modern processors support fine-grain measurement using sampling~\cite{dean1997profileme, drongowski2007instruction, guide2011intel, cuptipcsampling}, which can be used to study instruction statistics in applications quantitively.
Unique among GPU vendors, NVIDIA implements PC sampling on its GPUs to sample instructions and associate them with stall reasons.
Existing performance tools~\cite{shende2006tau, hpctoolkit, cudablamer-protools19, nsightcompute, january2015allinea} that utilize PC sampling only associate instruction samples with source lines of GPU code where the stalls occur but lack the ability to derive performance insight based on stall reasons.

To complement the aforementioned approaches, we propose GPA---a GPU performance advisor that suggests effective optimizations for GPU code, and evaluate GPA on a V100 GPU with the Rodinia benchmarks~\cite{che2009rodinia}, several larger application benchmarks, and a combustion application.
Guided by GPA, we improved the performance of the GPU kernels studied by 1.03x to 3.86x. 
This paper describes the design and implementation of GPA which consists of the following key components:
\begin{itemize}
\item An instruction blamer that attributes stalls to instructions that cause them;
\item Performance optimizers that match inefficiency patterns with optimization suggestions for lines, loops, and functions based on program structure, architectural features, measurement data, and control flow;
\item Performance estimators that model GPU execution using instruction samples to estimate speedups for each optimizer.
\end{itemize}

%

This rest of the paper is organized as follows.
Section~\ref{sec:Background and Motivation} reviews PC sampling and instruction format on NVIDIA's GPUs.
Section~\ref{sec:Overview} introduces the workflow of GPA.
Section~\ref{sec:Instruction Blamer} explains the details of GPA's instruction blamer.
Section~\ref{sec:Performance Optimizers and Estimators} describes the implementation of GPA's preformance optimizers and estimators.
Section~\ref{sec:Evaluation} describes the analysis and optimization of GPU kernels using GPA.
Section~\ref{sec:Case Studies} presents case studies of four 
larger codes, including a combustion application.
Section~\ref{sec:Related Work} reviews related work and distinguishes GPA.
Finally, Section~\ref{sec:Conclusion} summarizes our work and outlines our plans for future work.

\section{Background and Motivation}~\label{sec:Background and Motivation}

\begin{table*}[tp]
\caption{Dissection of the fields of ``\texttt{@P0 LDG.32 R0, [R2]}'' instruction.}
\centering
\footnotesize	
\begin{tabular}{|c|c|c|c|c|c|c|c|}
\hline
Wait Mask & Write Barrier & Read Barrier & Predicate   & Opcode & Modifiers & Destination Operands & Source Operands \\ \hline
\texttt{B0}   & \texttt{B1}     &                     & \texttt{P0} & \texttt{LDG}  & \texttt{32}    & \texttt{R0} & \texttt{R2, R3}  \\ \hline
\end{tabular}
\label{tab:instruction}
\end{table*}

In this section, we describe background necessary to understand our work and our motivation for developing GPA. In Section~\ref{subsec:PC Sampling}, we introduce a model of the PC sampling mechanism implemented in recent NVIDIA GPUs.
In Section~\ref{subsec:Instruction Format}, we describe the instruction format used by NVIDIA's GPUs, which is important for instruction dependency analysis.
In Section~\ref{subsec:Motivation Examples}, we show how raw PC sampling measurements are insufficient to provide insight for performance optimization.

\subsection{PC Sampling}~\label{subsec:PC Sampling}
\begin{figure}
\centering
\includegraphics[width=0.9\linewidth]{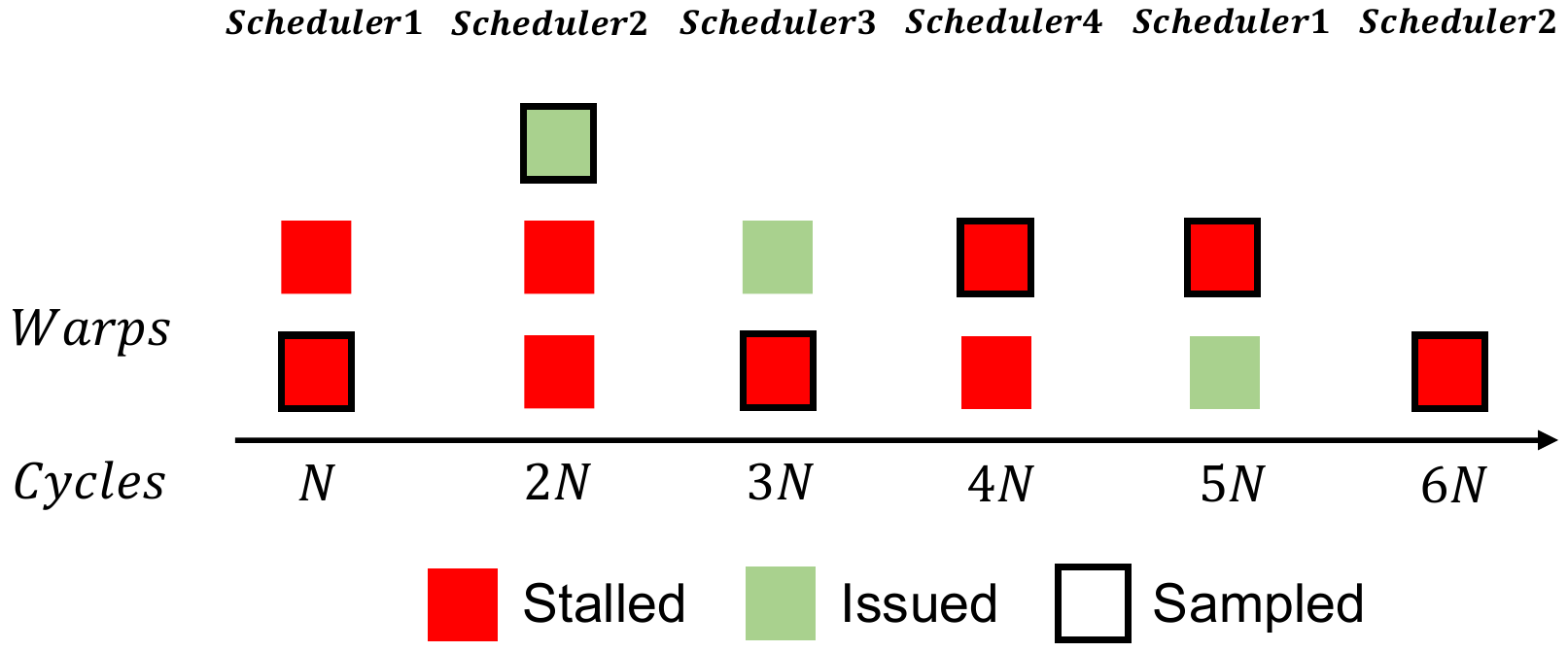}
\caption{A mental model of PC sampling on an SM of NVIDIA's V100 GPU.
Samples are taken every $N$ cycles.
Samples at $N$, 4$N$, and 6$N$ are latency samples, and others are active samples.
Samples at $N$, 3$N$, 4$N$, 5$N$, and 6$N$ are stall samples.}
\label{fig:pc sampling}
\end{figure}

NVIDIA's GPUs implement PC sampling to collect instruction samples.
One can use NVIDIA’s CUPTI API~\cite{cupti} to collect PC samples for GPU-accelerated applications.
Each streaming multi-processor (SM) in an NVIDIA GPU collects samples individually.
When a buffer used to collect samples is full on an SM, CUPTI merges samples from all SMs and transfers the samples to the CPU.

Each SM on an NVIDIA V100 has four warp schedulers, and each warp scheduler is assigned a number of active warps.
At the end of each sampling period, an SM records a sample for one of its warp schedulers and it cycles through its warp schedulers in a round-robin fashion. 
When a warp is sampled, two classes of samples are recorded: an \textit{active sample} when the warp scheduler is issuing an instruction and a \textit{latency sample} when no instruction is issuing.
For the instruction sampled, a stall reason (e.g., waiting for a value from memory) is recorded for the instruction, if any.
Consider Figure~\ref{fig:pc sampling} as an example.
There are 5 samples with a stall reason. We call them \textit{stall samples} or \textit{stalls} in the remaining sections.
Because there are three latency samples and three active samples, we estimate the stall ratio and the active ratio of the SM as $3/6$.
Assuming all SMs on the GPU have a similar workload, we estimate the stall ratio and the active ratio of the GPU kernel as $3/6$.

\subsection{Instruction Format}~\label{subsec:Instruction Format}
A fixed length instruction encoding is used on NVIDIA's GPUs. Pre-Volta GPUs use a 64-bit word for an instruction, but Volta and later architectures use a 128-bit word.
In this paper, we focus on the Volta architecture used in two of the top three supercomputers—Summit and Sierra.

Among the fields of a GPU instruction shown in Table~\ref{tab:instruction},
we focus on the following three key fields:

\begin{itemize}
\item \textbf{Wait Mask and Write/Read Barrier.}
Every GPU instruction has a control code~\cite{jia2018dissecting,zhang2017understanding} field that encodes information to guide the warp scheduler as it issues instructions, including stall cycles, yielding flag, and dependencies.
For each fixed latency instruction (e.g., most arithmetic instructions), the assembler sets stall cycles for the instruction to indicate how long the scheduler should wait before issuing the instruction.
For each variable latency instruction, the assembler associates write/read barrier indices with it, and associates instructions that depend on them a wait mask to create dependencies.

\item \textbf{Predicate}.
If an instruction's predicate field is set, the instruction is executed when the predicate evaluates as true.
There are both true and false predicate conditions: \texttt{Pi} is a true predicate condition, and \texttt{!Pi} is a false predicate condition, where $0 \leq i \leq 6$.
In Table~\ref{tab:instruction}, the \texttt{LDG} instruction is executed if \texttt{P0} is true.

\item \textbf{Opcode, Modifiers, and Operands}.
Each thread can use up to 255 32-bit regular registers ranging from \texttt{R0}-\texttt{R254}.
Opcode and modifiers together determine the length of operands used.
In Table~\ref{tab:instruction}, the \textit{32} modifier indicates each thread reads a 32-bit value from memory.
Moreover, because the data is loaded from global memory, which has a 64-bit address space, the source operand is a 64-bit value comprised of two registers---\texttt{R2} and \texttt{R3}.
\end{itemize}

\subsection{Motivating Examples}~\label{subsec:Motivation Examples}

We refer to a collection of instruction samples and their stall reasons as a raw PC sampling report from which we can measure the stall reasons of a kernel.
However, diagnosing the slowness of the kernel still requires interpretation of the measurement data to answer the following questions. 

\begin{itemize}
\item Which GPU instructions cause stalls?
\item How can we improve the performance by eliminating these stalls?
\item What is the estimated speedup for each potential optimization?
\end{itemize}

To illustrate the importance of analyzing stall reasons and associating them with optimizations, we analyze the \textit{hotspot} and the \textit{b+tree} examples in Rodinia benchmark. 

\begin{figure}[htbp]
\centering
\begin{lstlisting}[language=C++, caption={A hot loop in the \textit{hotspot} example}, label={lst:hotspot}]
for (int i = 0; i < iteration; i++) {
  temp_t[ty][tx] = | $\label{line:hotspot execution}$ |
    temp_on_cuda[ty][tx] + step_div_Cap * (
    power_on_cuda[ty][tx] + (temp_on_cuda[S][tx] +
    temp_on_cuda[N][tx] - 2.0 * temp_on_cuda[ty][tx]) *
    ...
}
\end{lstlisting}
\end{figure}

Listing~\ref{lst:hotspot} shows a hot loop of the \textit{hotspot} kernel.
The raw PC sampling report for this kernel indicates large execution latency stalls on Line~\ref{line:hotspot execution}, but it provides little information regarding where the stalls come from and what optimizations apply.
GPA attributes the latency to type conversion instructions that demote a 64-bit float to a 32-bit float.
Though all arrays are composed of 32-bit values, the compiler generates conversion instructions as a float constant multiplies a 32-bit float value.
GPA suggests specifying the type of the constant ($2.0$) as a 32-bit value to avoid conversion.
After applying the optimization, we achieved a 1.14$\times$ speedup.

\begin{figure}[htbp]
\centering
\begin{lstlisting}[language=C++, caption={A hot loop in the \textit{b+tree} example}, label={lst:b+tree}]
for (int i = 0; i < height; i++) {
  if ((knodesD[currKnodeD[bid]].keys[thid] <= startD[bid]) && | $\label{line:b+tree memory}$ |
      (knodesD[currKnodeD[bid]].keys[thid+1] > startD[bid])) 
  ...
  __syncthreads();| $\label{line:b+tree sync}$ |
}
\end{lstlisting}
\end{figure}

Listing~\ref{lst:b+tree} shows a costly loop in the b+tree code.
The raw PC sampling report shows high memory dependency stalls on Line~\ref{line:b+tree memory} but does not propose a suggestion to eliminate the bottleneck.
By analyzing the assembly code, GPA concludes that the distance between the load instructions and the instruction that consumes the loaded values is short.
Therefore, instructions in the path are not enough to hide the latency.
GPA suggests the users separate the subscripted loads from their uses 
by reordering code. We read the address of {\tt knodesD[currKnodeD[bid]].keys} for the next iteration before the synchronization on Line~\ref{line:b+tree sync} and obtained a 1.16$\times$ speedup.

Based on the analysis above, we conclude that pure PC sampling information is insufficient to guide optimizations.
To improve the quality of the analysis report, we analyze instruction dependencies to characterize stalls' causes.
Furthermore, we can associate the stalls with the program's structure to suggest code optimizations, such as loop unrolling, function inlining, and code reordering.

\begin{figure}[tbp]
\centering
\includegraphics[width=0.9\linewidth]{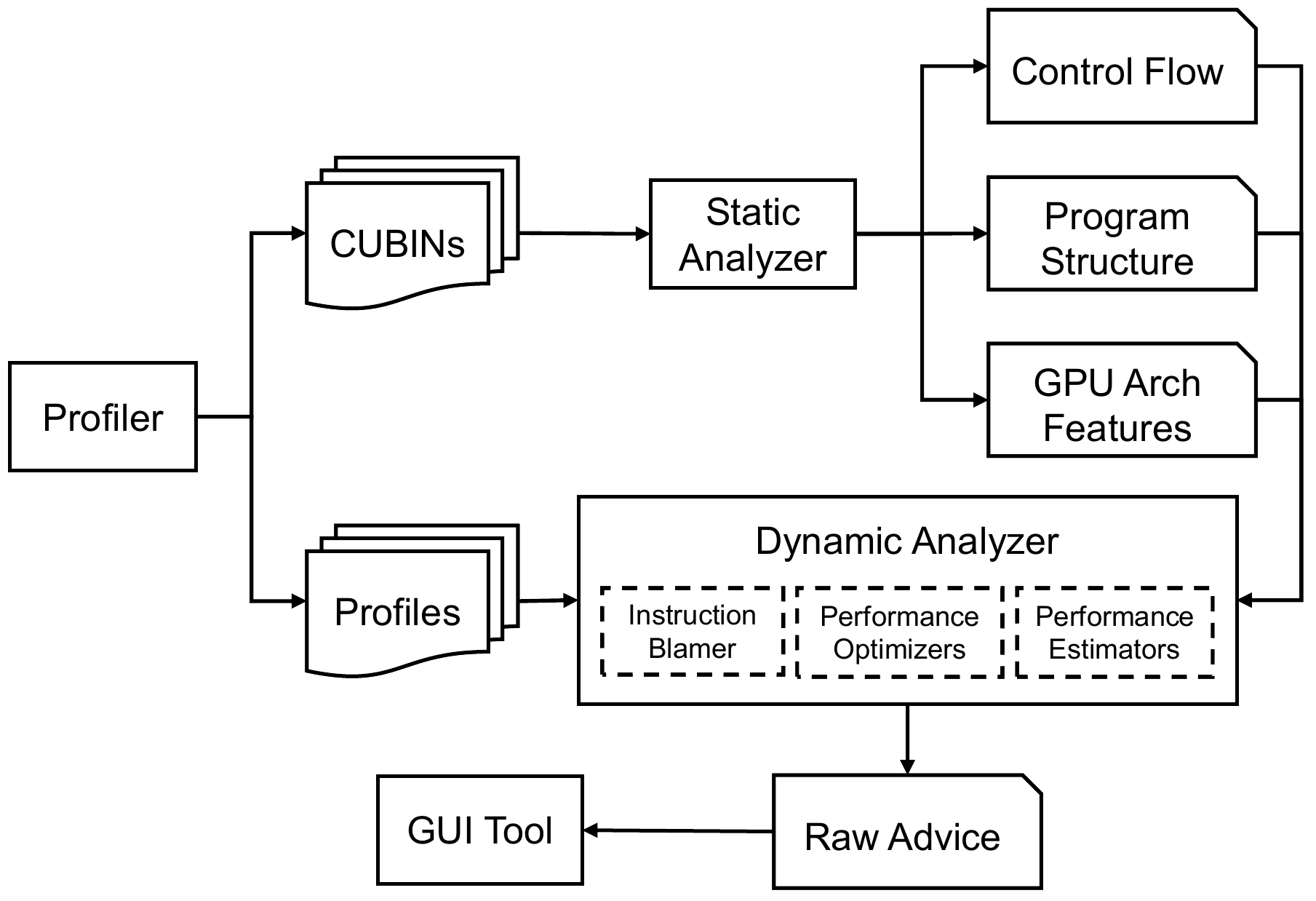}
\caption{Overview of GPA}
\label{fig:framework}
\end{figure}

\section{Overview}~\label{sec:Overview}

Figure~\ref{fig:framework} shows the workflow of GPA.
GPA uses a \textit{profiler} to collect PC samples and kernel launch statistics at runtime and attribute them to the calling context where the kernel is launched.
The profiler dumps the profiles and records CUDA binaries (CUBINs) for offline analysis.
GPA's \textit{static analyzer} analyzes CUBINs to recover static information which is ingested into the \textit{dynamic analyzer} with profiles to generate comprehensive \textit{raw advice}.

\paragraph{Static Analyzer}
In its static analyzer, GPA analyzes CUBINs to recover the following files:
\begin{itemize}
\sloppy
\item \textbf{Control flow graphs}. GPA employs NVIDIA's \texttt{nvdisasm} tool to decode instructions in CUBINs and dump raw control flow graphs.
We modify the raw control flow graphs by splitting super blocks into basic blocks and ingest the modified control flow graphs into Dyninst~\cite{dyninst} to analyze loop nests.
\item \textbf{Program structure}. A program structure file contains functions symbols, inline stacks, loop nests, and source line mappings.
According to each function symbol's visibility field, we annotate global functions and device functions.
We read DWARF information to parse information about inlined functions.
\item \textbf{Architectural features}.
Based on the architecture flag encoded in CUBINs, we fetch specific hardware configurations, such as instruction latencies, warp size, and register limitations for analysis in the later stages.
\end{itemize}

\paragraph{Dynamic Analyzer}

The \textit{dynamic analyzer} is comprised of three components, including an \textit{instruction blamer}, \textit{performance optimizers}, and \textit{performance estimators}.

We analyze each GPU kernel's launch context separately.
For each kernel invocation, the \textit{instruction blamer} uses backward slicing~\cite{cifuentes1997intraprocedural, srinivasan2016improved} to attribute stalls to the responsible instructions.
Based on the stall counts and GPA's static analysis results, each \textit{performance optimizer} attempts to match its optimization strategy to program regions that have high stall samples.
Guided by performance models, \textit{performance estimator}s estimate each optimizer's speedup based on the matched samples.
Finally, GPA generates an advice report that contains suggestions from its top optimizers sorted by their estimated speedups. 

In this paper, we focus on the implementation of GPA's dynamic analyzer, which tackles the following unique challenges: 
(1) It extends the backward slicing algorithm for special fields (e.g., barriers) of a GPU instruction to track dependencies among GPU instructions. 
(2) It attributes stalls to their sources accurately because it incorporates  pruning rules to cut down dependency sources.
(3) Without code annotation, it derives a general performance model to quantify the benefits of each GPU optimizer.

\paragraph{Utilization of GPA}
GPA is a command line tool that automates profiling and analysis stages.
Since GPA uses sampling-based profiles, users do not need to change their program source code.
To provide advice at the source line level, the only change required is adding compiler options to ensure that the compiler includes line mapping information in GPU binaries it generates.
Users apply optimizations according to the raw advice generated by GPA.
Today, GPA produces raw advice as ASCII text; however, its advice could be incorporated into a graphical user interface tool to analyze inefficient code regions and optimization suggestions.

\section{Instruction Blamer}~\label{sec:Instruction Blamer}

CUPTI associates stall reasons~\cite{cuptipcsampling} with instruction samples.
Among the stall reasons, memory dependency, synchronization, and execution dependency stalls are caused by the source instructions rather than the instructions that suffer from stalls.
Other stall reasons, such as memory throttling, are caused by instruction samples with the stall.
To further characterize program bottlenecks with memory dependency, synchronization, and execution dependency stalls, we developed an instruction blamer that attributes stalls to the source instructions.

We first use backward slicing to analyze every instruction's def-use chain in the control flow graph.
According to the def-use chain and measurement data, we build an instruction dependency graph where each node is an instruction, annotated with its stalls, and each edge represents a def-use relation.
Since not all edges cause stalls, we prune edges according to several heuristic rules.
In the end, we apportion the stalls to its incoming edges based on the number of issued instructions and the length of each edge.

\paragraph{Backward slicing}

We target intra function backward slicing~\cite{cifuentes1997intraprocedural} for GPU instructions because instructions in the same function cause most stalls.
We find a stalled instruction's immediate dependency sources because transitive dependencies are unlikely to cause the stalls.
According to Table~\ref{tab:instruction}, several fields of a GPU instruction impact instruction dependencies, including operands, barriers, and predicate.
We can begin with a traditional backward slicing algorithm for CPU instructions to analyze GPU operands, but barriers and predicates need special processing.

\begin{figure}[t]
\centering
\includegraphics[width=0.4\linewidth]{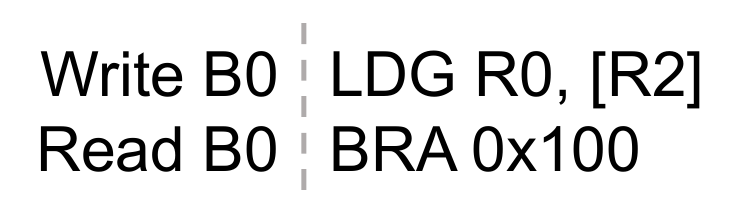}
\caption{An example of barrier register dependency}
\label{fig:barrier}
\end{figure}

\textit{Virtual barrier registers:}
We define six available barrier indices as six virtual barrier registers \texttt{B0}-\texttt{B5}.
A write/read barrier index association can be represented as a write operation to one or more barrier registers.
Likewise, we treat a wait mask association as a read of barrier registers.
In this way, dependencies caused by barrier indices can be identified through def-use chains of the virtual barrier registers.
It is worth noting that barriers can be set even if there is no dependency between regular registers.
Take Figure~\ref{fig:barrier} as an example, the \texttt{LDG} instruction loads a value to \texttt{R0} and writes barrier \texttt{B0}, and the \texttt{BRA} instruction does not consume \texttt{R0} but still reads \texttt{B0}.
Observed memory dependency stalls on the \texttt{BRA} instruction should be attributed to the \texttt{LDG} instruction.

\textit{Predicated instructions:}
Immediate dependency sources are not only the first \textit{def} instruction of each of its operands on the search path.
Consider Figure~\ref{fig:backward slicing} as an example, suppose we observe a stall at the \texttt{IADD} instruction, which does not have a predicate;
because the \texttt{LDG} instruction is executed only if \texttt{P0} is true, it is possible that the stall comes from the \texttt{LDC} instruction earlier in the path, which is executed only if \texttt{P0} is false.
Therefore, the backward slicing search should proceed until the predicates of \textit{def} instructions on the path cover all conditions.

Let $P$ be the union of \textit{def} instructions' predicates on the path.	
$P = \cup p$, where $p \in \{p_{i}\} \cup \{!p_{i}\} \cup \{\_\}$, and $\{ p_{i} \} \cup \{ !p_{i} \}= \{ \_ \}$, for $0\leq i \leq 6$.	
$\_$ is a special predicate that covers both true and false predicates.
An instruction without a predicate has the same semantic as $\_$.
We say $P$ \textit{contains} $p'$ iff $p' \in P$ or $\_ \in P$.
The backward slicing search proceeds until the union of \textit{def} instructions' predicates on the search path ($P$) \textit{contains} the predicate of the \textit{use} instruction ($p'$).

\begin{figure*}[htbp]
\begin{subfigure}{0.5\linewidth}
\centering
\includegraphics[width=0.45\linewidth]{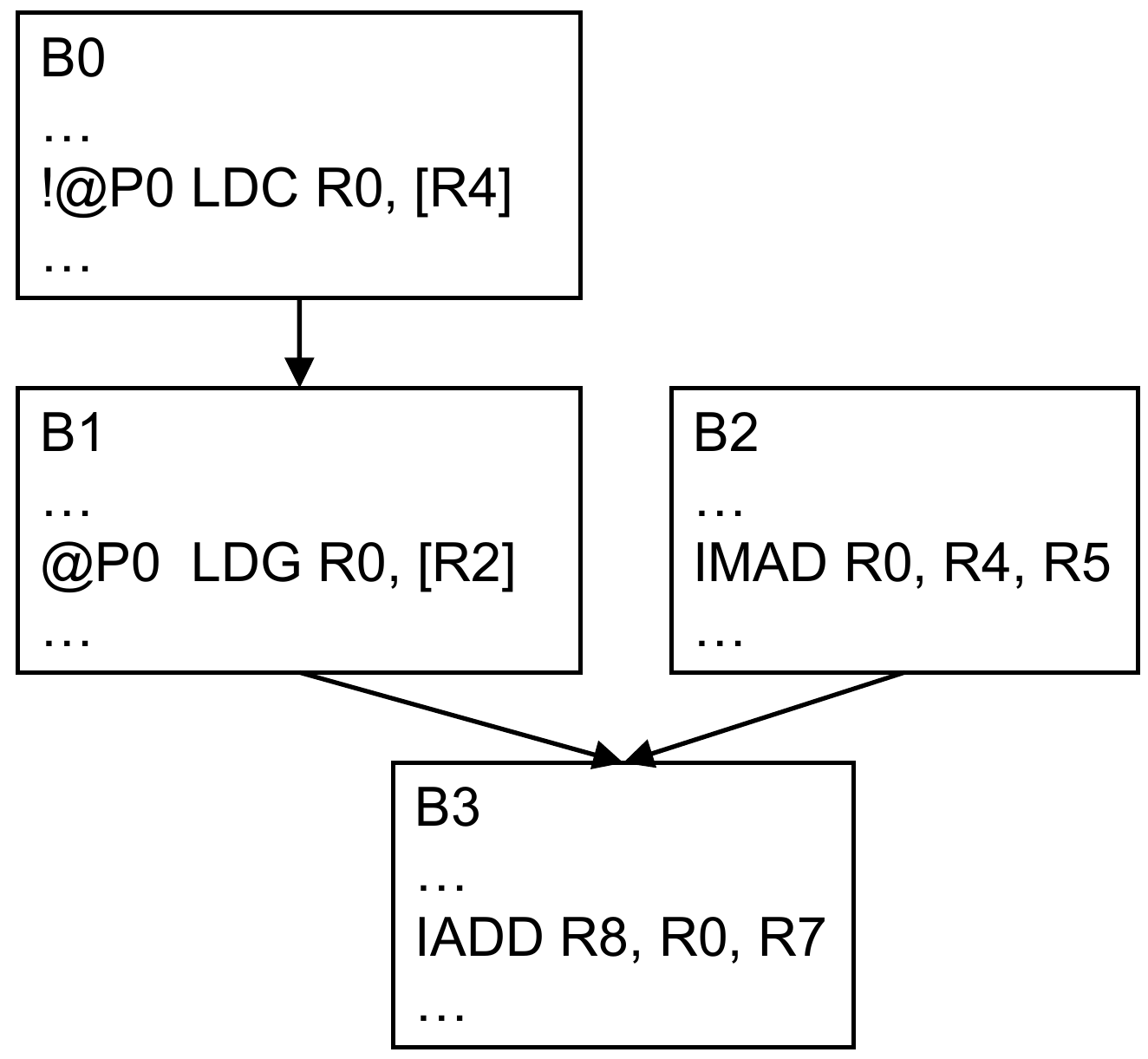}
\caption{Backward slicing}
\label{fig:backward slicing}
\end{subfigure}
~
\begin{subfigure}{0.5\linewidth}
\centering
\includegraphics[width=0.72\linewidth]{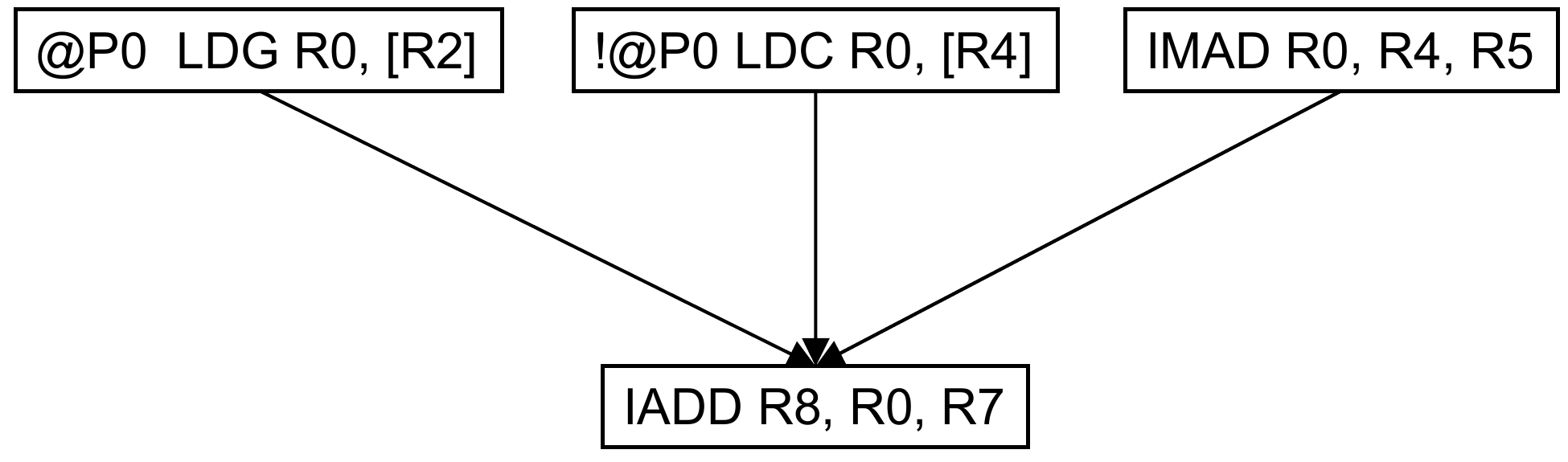}
\caption{Construct a dependency graph}
\label{fig:construct a dependency graph}
\end{subfigure}
\newline
\begin{subfigure}{0.5\linewidth}
\centering
\includegraphics[width=0.66\linewidth]{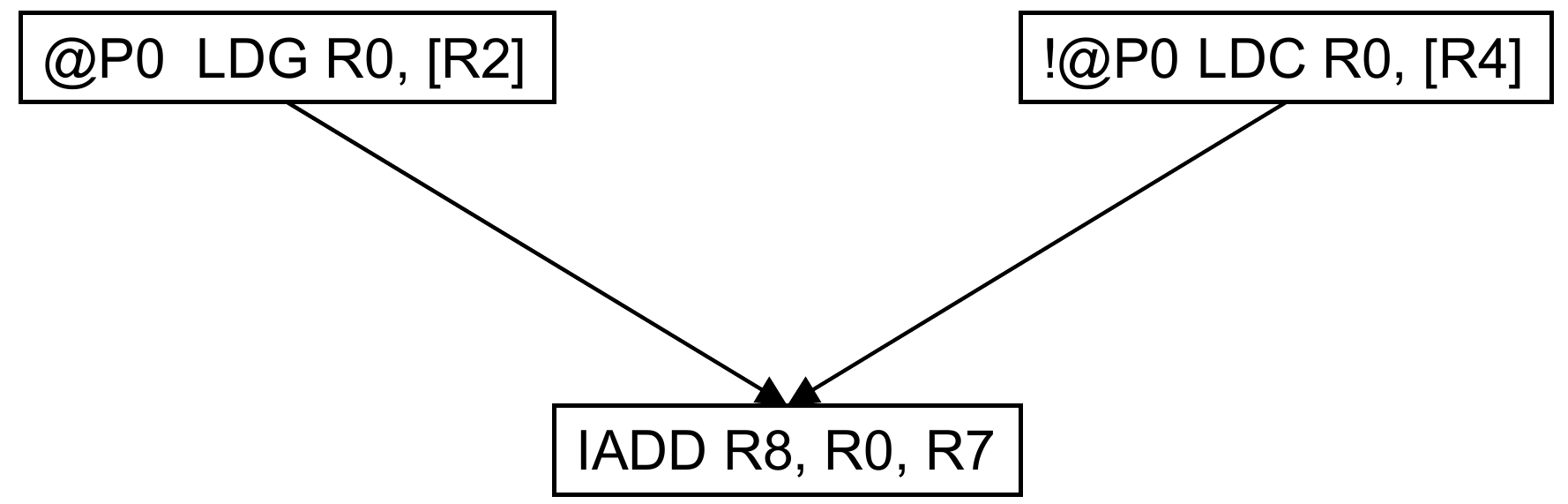}
\caption{Prune cold edges}
\label{fig:prune cold edges}
\end{subfigure}
~
\begin{subfigure}{0.5\linewidth}
\centering
\includegraphics[width=0.66\linewidth]{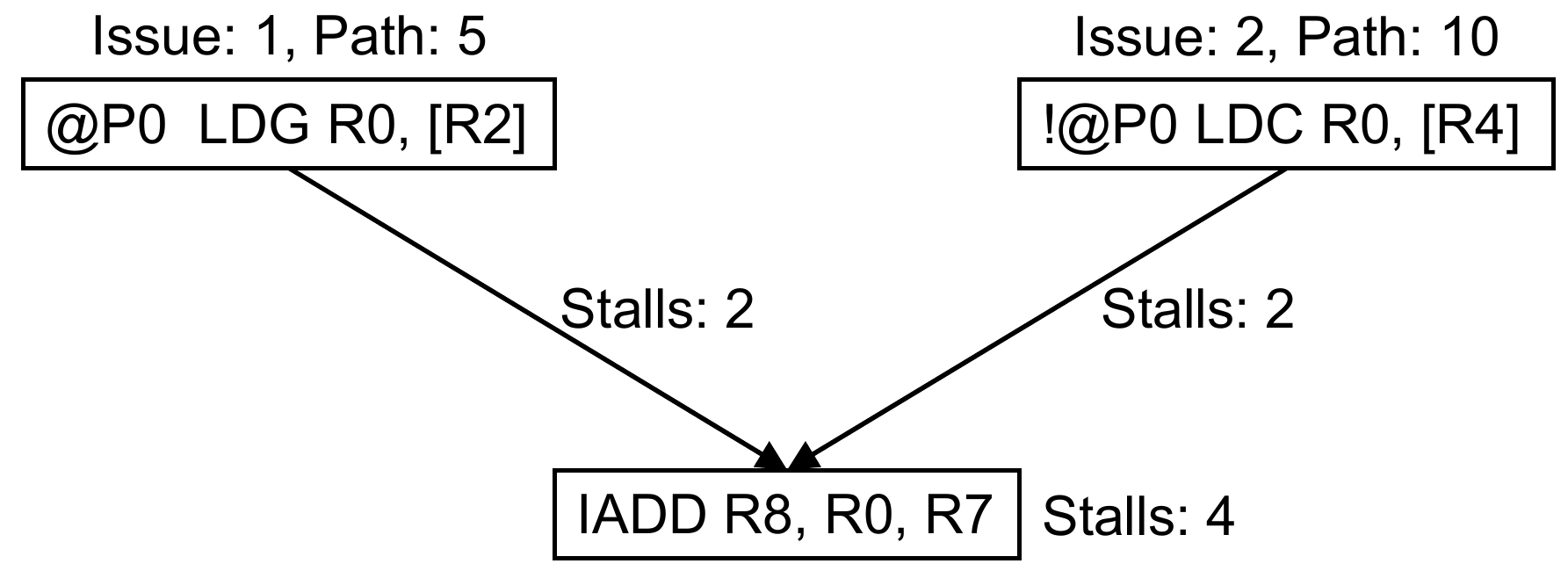}
\caption{Apportion stalls}
\label{fig:apportion stalls}
\end{subfigure}
\caption{Steps to attribute stalls of  the \texttt{IADD} instruction}
\label{fig:instruction blamer}

\vspace{-1ex}
\end{figure*}

\paragraph{Construct a dependency graph} We build an instruction dependency graph from the def-use chains of collected instruction samples.
For simplicity, in Figure~\ref{fig:construct a dependency graph} we only demonstrate memory dependency.
Each node represents an instruction, and each edge represents a def-use relation associated with \texttt{R0}.

\paragraph{Prune cold edges} Not all the dependent edges cause stalls.
If an edge does not trigger stalls, we call it a ``cold edge'' and use the following three rules to prune it.

\begin{enumerate}
\item \textbf{Opcode based pruning}. Memory dependency stalls are attributed to memory instructions only. Synchronization dependency stalls are attributed to synchronization instructions only.
\item \textbf{Dominator based pruning}. For every edge $e$ from node $i$ to $j$ in a dependency graph, we remove $e$ if there is a non-predicate instruction $k$ uses the same operands that $i$ defines and $j$ uses, and $k$ is in every path from $i$ to $j$ in the control flow graph because we would have observed stalls at $k$ rather than $j$ if $i$ caused any stalls.
\item \textbf{Instruction latency based pruning}. For every edge $e$ from node $i$ to $j$ in a dependency graph, we remove $e$ if the number of instructions in every path from $i$ to $j$ in the control flow graph is greater than the $latency$ of $i$. 
\end{enumerate}

For fixed latency instructions, we can use microbenchmarking~\cite{jia2018dissecting} for their latencies; for variable latency instructions, we use their upper bounds for pruning.
For instance, we use the TLB miss latency as the upper bound latency of global memory instructions.

According to the opcode pruning rule, we prune the edge from \texttt{IMAD} to \texttt{IADD} in Figure~\ref{fig:construct a dependency graph} to obtain the dependency graph in Figure~\ref{fig:prune cold edges} because an \texttt{IMAD} instruction cannot cause memory dependency stalls.

\paragraph{Attribute stalls} After pruning cold edges, there are still some nodes that have multiple incoming edges.
To measure the stalls caused by each edge, we use the following two heuristics.
\begin{enumerate}
\item Apportion the stalls based on each incoming node's issued samples. The more the issued samples, the more stalls are blamed to the instruction. 
\item Apportion the stalls based on the number of instructions in paths. The longer the path, the less stalls are blamed on the \textit{def} instruction. If an instruction $i$ has multiple paths to instruction $j$ in a control flow graph, we use the longest one.
\end{enumerate}

Finally, we associate the stalls of each dependency source ($S_{i}$) by apportioning the stalls of the observed instruction ($S_{j}$) using Equation~\ref{eq:apportion stalls}, where $\mathcal{R}_{i}^{issue}$ is the ratio of each incoming node calculated by heuristic (1), and $\mathcal{R}_{i}^{path}$ denotes the ratio of each dependency source $i$ calculated by heuristic (2).

\begin{equation}~\label{eq:apportion stalls}
S_{i} = \frac{\mathcal{R}_{i}^{path} \times \mathcal{R}_{i}^{issue}}{\sum\limits_{k \in incoming(j)}{\mathcal{R}_{k}^{path} \times \mathcal{R}_{k}^{issue}}} \times S_{j}
\end{equation}

Figure~\ref{fig:apportion stalls} shows the apportioned stalls using the above heuristics. While the \texttt{LDC} instruction has twice the issued samples of the \texttt{LDG} instruction, the number of path samples from \texttt{LDC} to \texttt{IADD} is also twice that of \texttt{LDG} to \texttt{IADD}. Thus, we assign each dependency source the same number of samples.

Without loss of generality, the above heuristics and equation also apply for apportioning latency samples.

\begin{figure}[htbp]
\begin{subfigure}{\linewidth}
\centering
\includegraphics[width=0.8\linewidth]{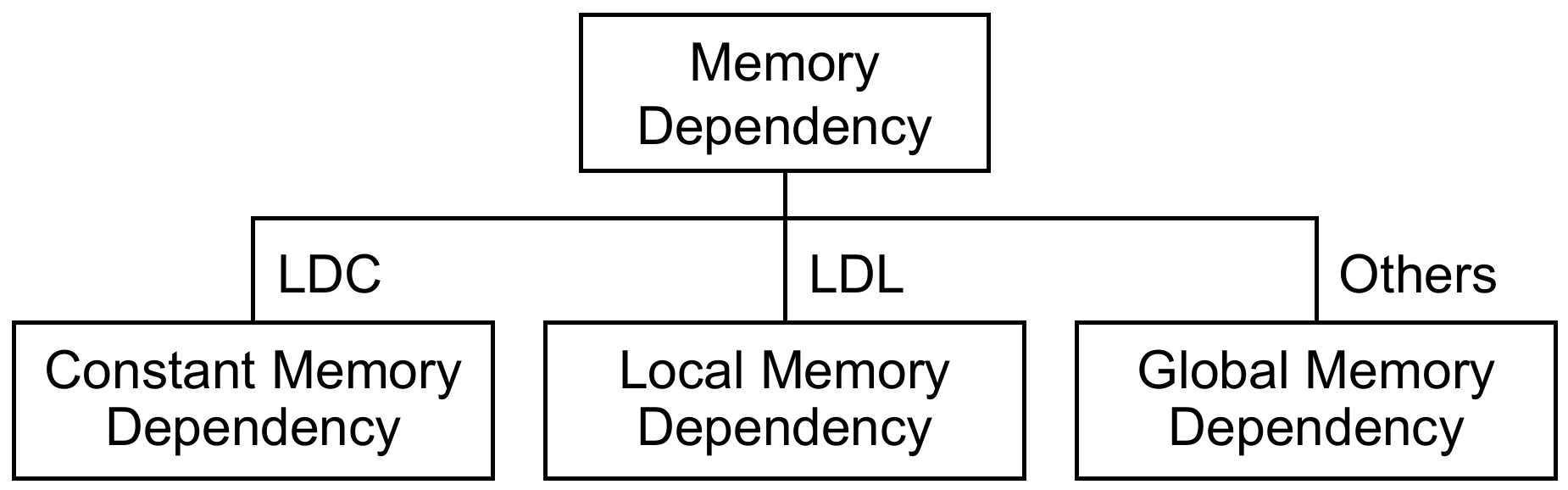}
\caption{Memory dependency}
\end{subfigure}
\begin{subfigure}{\linewidth}
\centering
\includegraphics[width=0.8\linewidth]{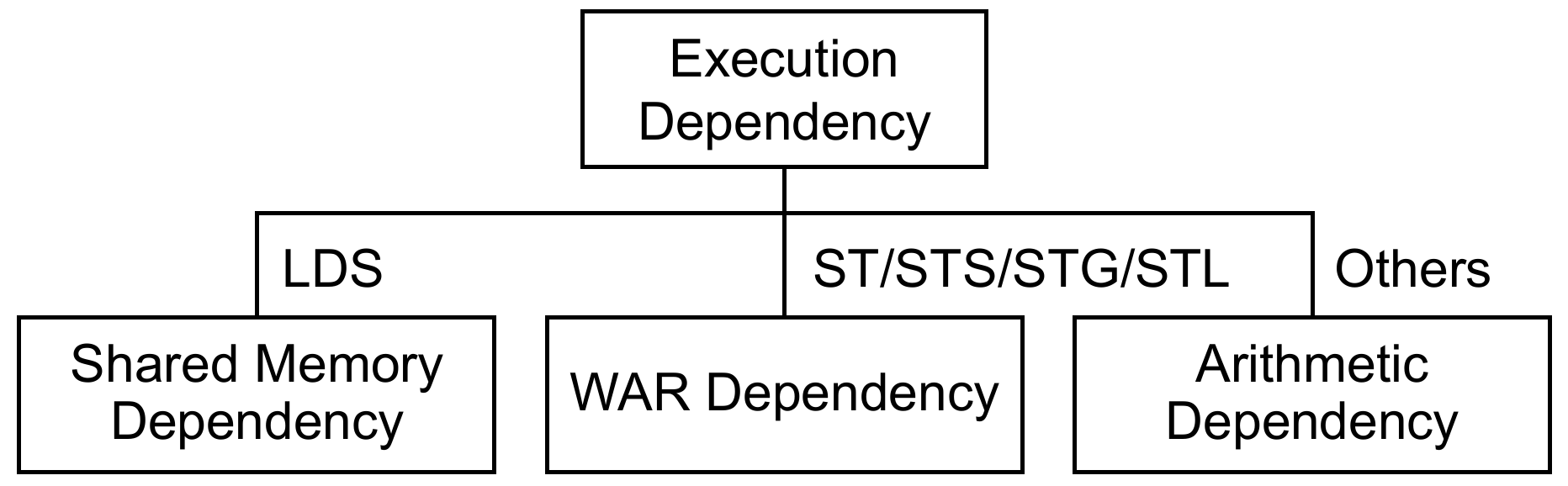}
\caption{Execution dependency}
\end{subfigure}
\caption{Classification of detailed dependency stall reasons}
\label{fig:dependency stalls}
\end{figure}

After attributing stalls to their sources, we further classify the stall reasons for execution and memory dependencies according to the opcode of each source instruction.
As shown in Figure~\ref{fig:dependency stalls}, we categorize memory dependency as local memory, constant memory, and global memory dependencies.
Knowing where local memory stalls occur is important for register pressure analysis because it often indicates register spills.
Likewise, we classify execution dependency as shared memory, arithmetic, and write-after-read (WAR) dependencies.
WAR dependency happens when a variable latency \textit{def} instruction reads a value from a register, and the \textit{use} instruction writes the same register.

\section{Performance Optimizers and Estimators}~\label{sec:Performance Optimizers and Estimators}

This section describes the implementation of performance optimizers and estimators.

\subsection{Performance Optimizers}

Performance optimizers take program structure and the analysis result from the instruction blamer.
Each optimizer encodes rules to calculate matching stalls.
In this way, we lift the job of associating stalls with optimizations from users to the advisor.

\begin{table}[htbp]
\centering
\caption{A brief description of GPU optimizers in GPA.}
\scriptsize
\begin{tabular}{|c|c|c|}
\hline
\multicolumn{3}{|c|}{Code Optimizers}                                                                                                                 \\ \hline
 & \multicolumn{2}{c|}{Stall Elimination}                                                                                                             \\ \hline
 & Register Reuse               & \begin{tabular}[c]{@{}c@{}}Match memory dependency stalls\\ of local memory read/write instructions\end{tabular}    \\ \hline
 & Strength Reduction           & \begin{tabular}[c]{@{}c@{}}Match execution dependency stalls of\\ long latency arithmetic instructions\end{tabular} \\ \hline
 & Function Split               & Match instruction fetch stalls                                                                                      \\ \hline
 & Fast Math                    & Match stalls in CUDA math functions                                                                                 \\ \hline
 & Warp Balance                 & Match warp synchronization stalls                                                                                   \\ \hline
 & Memory Transaction Reduction & Match global memory throttling stalls                                                                               \\ \hline
 & \multicolumn{2}{c|}{Latency Hiding}                                                                                                                \\ \hline
 & Loop Unrolling               & \begin{tabular}[c]{@{}c@{}}Match global memory and execution\\ dependency stalls in loops\end{tabular}              \\ \hline
 & Code Reordering              & \begin{tabular}[c]{@{}c@{}}Match global memory and execution\\ dependency stalls\end{tabular}                       \\ \hline
 & Function Inlining            & \begin{tabular}[c]{@{}c@{}}Match stalls in device functions\\ and their call sites\end{tabular}                     \\ \hline
\multicolumn{3}{|c|}{Parallel Optimizers}                                                                                                             \\ \hline
 & Block Increase               & \begin{tabular}[c]{@{}c@{}}Match if the number of blocks\\ is less than the number of SMs\end{tabular}              \\ \hline
 & Thread Increase              & \begin{tabular}[c]{@{}c@{}}Match if occupancy is limited by\\ the number of threads per block\end{tabular}          \\ \hline
\end{tabular}
\label{tab:optimizers}
\end{table}

We classify the available performance optimizers in GPA in Table~\ref{tab:optimizers}.
At a high level, we have parallel and code optimizers.
Parallel optimizers check if we can increase the parallelism level to improve performance.
For instance, the \textit{Block Increase} optimizer investigates the potential of increasing the number of blocks.
Code optimizers check if we can adjust code to improve the performance.
Based on optimization methods, we further categorize the code optimizers as stall elimination and latency hiding optimizers.
Stall elimination optimizers provide suggestions to reduce stalls; latency hiding optimizers suggest rearranging issue orders to overlap stall latency.

Each optimizer maintains a workflow to match instruction samples.
The \textit{Loop Unrolling} optimizer, for example, iterates through all the latency samples.
It records a latency sample if it has either a memory dependency stall or an execution dependency stall, and the def and the use instructions are within the same loop.
The optimizer suggests using \texttt{pragma unroll} annotation or manual unrolling for loops where the compiler fails to unroll automatically.

\subsection{Performance Estimators}

With performance optimizers, we associate optimization methods with stalls, whereas it is still unclear which methods have a better effect in terms of the given measurement data, program structure, and the underlying GPU architecture.
Performance estimators take the matched stalls as input and estimate the speedups by modeling the GPU's execution.
The optimizers with top estimated speedups output their suggestions to the performance advice report.
According to the categories of optimizers, we classify estimators as code optimization estimators and parallel optimization estimators.

\subsubsection{Code Optimization Estimators}

We first model the effect of the stall elimination optimizers.
Suppose the total of number samples for a GPU kernel is $T$, and the matched samples for an optimizer is $M$.
Stall elimination optimizers assume we at best eliminate all the stalls by modifying the code.
We use Equation~\ref{eq:stall elimination} to estimate the speedup of stall elimination optimizers $\mathcal{S}^{e}$.

\begin{equation}~\label{eq:stall elimination}
\mathcal{S}^{e} = \frac{T}{T - M}
\end{equation}

Latency hiding optimizers suppose we can at best eliminate latency samples by modifying code.
Therefore, we can use Equation~\ref{eq:latency hiding} to estimate the speedup of latency hiding optimizers $\mathcal{S}^{h}$, where $M^{L}$ is the number of matched latency samples.

\begin{equation}~\label{eq:latency hiding}
\mathcal{S}^{h} = \frac{T}{T - M^{L}}
\end{equation}

\begin{figure}[htbp]
\centering
\includegraphics[width=0.35\linewidth]{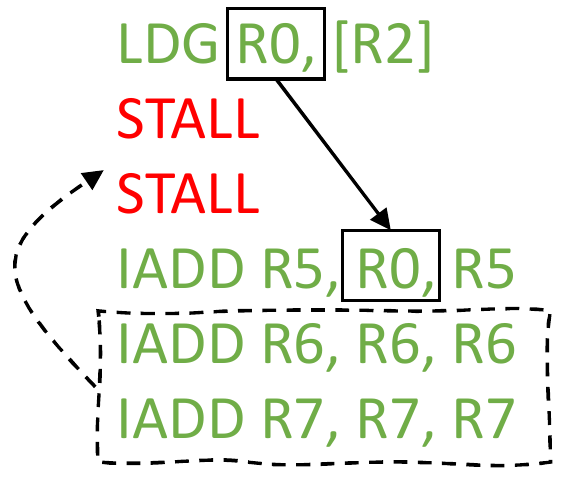}
\caption{The mental model of latency hiding optimizers.
Green code represents active samples, and red code represents latency samples.
Latency hiding optimizers consider the effect of moving the code enclosed in dashed lines to fill stall slots.}
\label{fig:reorder}
\end{figure}

Equation~\ref{eq:latency hiding} models the execution at the kernel level.
In practice, however, not all $M^{L}$ can be eliminated by rearranging code.
Figure~\ref{fig:reorder} explains the mental model of latency hiding optimization.
We derive Equation~\ref{eq:latency hiding improve} to refine the estimate of $\mathcal{S}^{h}$, where $A$ denotes the total number of active samples.

\begin{equation}~\label{eq:latency hiding improve}
\mathcal{S}^{h} = \frac{T}{T - Min(A, M^{L})}
\end{equation}

We prove that the upper bound of $\mathcal{S}^{h}$ is two. We use $L$ to denote the total number of latency samples, and $T=A + L$.

\begin{theorem}~\label{th:speedup}
The speedup upper bound of latency hiding optimizations is 2$\times$.
\end{theorem}
\begin{proof}
\begin{itemize}
\item If $Min(A, M^{L}) = A$. $\frac{T}{T - A} = \frac{L+A}{(L+A) - A} = 1 + \frac{A}{L}$. \\
Because $A\leq M^{L} \leq L, \frac{T}{T - Min(A, M_{L})} \leq 2$.
\item If $Min(A, M^{L}) = M^{L}$. $\frac{T}{T-M^{L}} = \frac{1}{1 - \frac{M^{L}}{T}} = \frac{1}{1 - \frac{M^{L}}{A + L}}$.\\
Because $L \geq M^{L}$ and $A \geq M^{L}$, $\frac{M^{L}}{A+L} \leq \frac{1}{2}$.\\
Then $\frac{T}{T - Min(A, M^{L})} \leq 2$.
\end{itemize}
\end{proof}

\paragraph{Scope Analysis}

We observe that optimizations such as loop unrolling only arrange code for a specific scope so that only the active samples within the scope can be used to reduce latency samples.
Based on this limitation, we propose Equation~\ref{eq:latency hiding scope} to analyze optimization scopes representing loops and functions.
$\mathcal{S}^{h}_{l}$ indicates the speedup for a specific scope $l$, and $M^{L}_{l}$ is the matched latency samples for a scope $l$.

\begin{equation}~\label{eq:latency hiding scope}
\mathcal{S}^{h}_{l} = \frac{T}{T - Min(\sum\limits_{l' \in nested(l)} A_{l'}, M^{L}_{l})}
\end{equation}

Suppose we have a loop \textit{loop1} nested in another loop \textit{loop2}, the speedup of of \textit{loop2} is bounded by the active samples of \textit{loop2} and \textit{loop1} according to Equation~\ref{eq:latency hiding scope}.

\subsubsection{Parallel Optimization Estimator}

Parallel optimizers adjust the number of blocks and threads within each block to change the parallelism level.
To estimate the effect of adjusting blocks and threads, we take into account each warp scheduler's change of active warps--$\mathcal{C}_{W}$ (Equation~\ref{eq:warp ratio}) and change of issue rate---$\mathcal{C}_{\mathcal{I}}$ (Equation~\ref{eq:issue ratio}) .

For instance, by increasing the number of blocks, we reduce the active warps per scheduler and $\mathcal{C}_{W}$ is less than one.
If the number of threads of each block is reduced, the rate that a warp scheduler is issuing is reduced, and $\mathcal{C}_{\mathcal{I}}$ is less than one.

\begin{equation}~\label{eq:warp ratio}
\mathcal{C}_{W} = \frac{W_{new}}{W}
\end{equation}

\begin{equation}~\label{eq:issue ratio}
\mathcal{C}_{\mathcal{I}} = \frac{\mathcal{I}_{new}}{\mathcal{I}}
\end{equation}

Assuming every warp scheduler's issue rate is the same across different SMs, we derive Equation~\ref{eq:parallel active} and Equation~\ref{eq:parallel active new} to calculate $\mathcal{I}$ and $\mathcal{I}_{new}$ respectively, where $R_{I}$ is the ratio of issued samples among all samples.
A warp scheduler is issuing if at least one warp on the scheduler is ready to issue an instruction.

\begin{equation}~\label{eq:parallel active}
\mathcal{I} = 1 - (1 - R_{I})^{W}
\end{equation}

\begin{equation}~\label{eq:parallel active new}
\mathcal{I}_{new} = 1 - (1 - R_{I})^{{W}_{new}}
\end{equation}

\begin{equation}~\label{eq:parallel speedup}
\mathcal{S}^{p} = \frac{1}{\mathcal{C}_{W}} \times \mathcal{C}_{\mathcal{I}} \times f
\end{equation}

Based on $\mathcal{C}_{W}$ and $\mathcal{C}_{\mathcal{I}}$, we estimate the speedup of parallel optimizations ($\mathcal{S}^{p}$) using Equation~\ref{eq:parallel speedup}, where $f$ is a factor that varies between optimizers.
Some optimizers may assume there is no pipeline, memory throttle, and no select stall if we reduce the number of active warps per block to a certain number (e.g., less than the number of schedulers per SM).

\section{Evaluation}~\label{sec:Evaluation}
\begin{table*}[htbp]
\centering
\caption{Achieved speedups averaged among ten runs. We improved each code according to the suggestion provided by GPA. Estimate error is computed by $\frac{|Estimated\ Speedup - Achieved\ Speedup|}{Achieved\ Speedup}\times 100\%$}
\scriptsize
\begin{tabular}{|c|c|c|c|c|c|c|}
\hline
\textbf{Application}   & \textbf{Kernel}               & \textbf{Optimization}        & \textbf{Original} & \textbf{Achieved Speedup} & \textbf{Estimated Speedup} & \textbf{Error} \\ \hline
rodinia/backprop       & bpnn\_layerforward\_CUDA      & Warp Balance                 & 17.26$\pm$0.21us           & 1.15$\pm$0.03$\times$            & 1.14$\times$                     & 1\%            \\ \hline
rodinia/backprop       & bpnn\_layerforward\_CUDA      & Strength Reduction           & 15.06$\pm$0.13us           & 1.21$\pm$0.01$\times$            & 1.24$\times$                     & 2\%            \\ \hline
rodinia/bfs            & Kernel                        & Loop Unrolling               & 567.14$\pm$2.04us          & 1.12$\pm$0.01$\times$            & 1.59$\times$                     & 42\%           \\ \hline
rodinia/b+tree         & findRangeK                    & Code Reorder                 & 51.59$\pm$0.39us           & 1.16$\pm$0.1$\times$            & 1.28$\times$                     & 10\%           \\ \hline
rodinia/cfd            & cuda\_compute\_flux           & Code Reorder                & 190.81$\pm$2.98ms          & 1.60$\pm$0.02$\times$            & 1.68$\times$                     & 5\%            \\ \hline
rodinia/gaussian       & Fan2                          & Thread Increase              & 105.76ms          & 3.58$\pm$0.15$\times$            & 3.32$\times$                     & 7\%           \\ \hline
rodinia/heartwall      & kernel                        & Loop Unrolling               & 94.06$\pm$2.79ms           & 1.17$\pm$0.03$\times$            & 1.18$\times$                     & 1\%            \\ \hline
rodinia/hotspot        & calculate\_temp               & Strength Reduction           & 11.92$\pm$0.08us           & 1.14$\pm$0.01$\times$            & 1.09$\times$                     & 4\%            \\ \hline
rodinia/huffman        & vlc\_encode\_kernel\_sm64huff & Warp Balance                 & 123.40$\pm$0.28us          & 1.07$\times$           & 1.17$\times$                     & 9\%            \\ \hline
rodinia/kmeans         & kmeansPoint                   & Loop Unrolling               & 784.41$\pm$4.81us          & 1.11$\pm$0.01$\times$            & 1.20$\times$                     & 8\%            \\ \hline
rodinia/lavaMD         & kernel\_gpu\_cuda             & Loop Unrolling               & 4.04$\pm$0.04ms            & 1.11$\pm$0.03$\times$            & 1.12$\times$                     & 1\%            \\ \hline
rodinia/lud            & lud\_diagonal                 & Code Reorder                 & 218.33$\pm$0.11us          & 1.41$\times$            & 1.48$\times$                     & 5\%            \\ \hline
rodinia/myocyte        & solver\_2                     & Fast Math                    & 308.55$\pm$6.87ms          & 1.22$\pm$0.03$\times$            & 1.13$\times$                     & 7\%            \\ \hline
rodinia/myocyte        & solver\_2                     & Function Splitting            & 258.09$\pm$0.27ms          & 1.02$\pm$0.01$\times$            & 1.01$\times$                     & 1\%            \\ \hline
rodinia/nw             & needle\_cuda\_shared\_1       & Warp Balance                 & 839.11$\pm$0.80us          & 1.07$\pm$0.01$\times$            & 1.09$\times$                     & 2\%            \\ \hline
rodinia/particlefilter & likelihood\_kernel            & Block Increase               & 2.05$\pm$0.02ms            & 1.75$\pm$0.01$\times$            & 1.92$\times$                     & 10\%            \\ \hline
rodinia/streamcluster  & kernel\_compute\_cost         & Block Increase               & 20.73$\pm$0.28ms           & 1.52$\pm$0.03$\times$            & 1.35$\times$                     & 11\%            \\ \hline
rodinia/sradv1         & reduce                        & Warp Balance                 & 1.94$\pm$0.28ms            & 1.02$\pm$0.01$\times$            & 1.10$\times$                     & 8\%           \\ \hline
rodinia/pathfinder     & dynproc\_kernel               & Code Reorder                 & 94.60$\pm$0.68us           & 1.04$\pm$0.01$\times$            & 1.31$\times$                     & 26\%           \\ \hline
Quicksilver            & CycleTracking\_Kernel           & Function Inlining            & 50.48$\pm$0.70s             & 1.14$\pm$0.01$\times$            & 1.18$\times$                     & 4\%            \\ \hline
Quicksilver            & CycleTracking\_Kernel           & Register Reuse               & 50.07$\pm$0.86s             & 1.02$\pm$0.01$\times$            & 1.04$\times$                     & 2\%            \\ \hline
ExaTENSOR              & tensor\_transpose             & Strength Reduction           & 5.60$\pm$0.02ms            & 1.11$\pm$0.01$\times$            & 1.06$\times$                     & 5\%            \\ \hline
ExaTENSOR              & tensor\_transpose             & Memory Transaction Reduction & 5.07$\pm$0.01ms            & 1.03$\times$            & 1.05$\times$                     & 2\%            \\ \hline
PeleC	& react\_state &	Block Increase & 121.60$\pm$1.05s & 1.21$\pm$0.01$\times$ & 1.23$\times$ & 2\% \\ \hline
Minimod	& target\_pml\_3d	  & Fast Math	   &  88.88$\pm$1.14ms & 1.03$\pm$0.01$\times$ & 1.09$\times$ & 6\% \\ \hline
Minimod	& target\_pml\_3d	  & Code Reorder   & 85.99$\pm$0.06ms & 1.04$\pm$0.01$\times$ & 1.05$\times$ & 1\% \\ \hline \hline
average & & & & 1.22$\times$ & 1.26$\times$ & 4.1\% \\ \hline
\end{tabular}
\label{tab:speedup}
\end{table*}

We evaluated GPA on an x86\_64 system with two Intel E5-2695 processors and a single NVIDIA Volta v100 GPU. 
The following system software are used: Linux 3.10.0, NVIDIA CUDA Toolkit 11.0.194, NVIDIA Driver 450.51.06, and GCC 7.3.0.
We evaluated GPA on Rodinia benchmarks and applications described below:

\begin{itemize}
\item Quicksilver~\cite{quicksilver} is a proxy application that solves a dynamic Monte Carlo particle transport problem.
Quicksilver has a single large kernel that invokes many device functions consisting of thousands of lines of code.
We studied Quicksilver with its default input.
\item ExaTENSOR~\cite{exatensor} is a library for large-scale numerical tensor algebra.
We studied its tensor transpose kernel using a large six-dimensional tensor.
\item PeleC~\cite{pelec} is an application for reacting flows using adaptive-mesh compressible hydrodynamics.
We studied PeleC using its default input.
\item Minimod~\cite{jie2020minimod} is a benchmark application for seismic modeling.
We analyzed its higher-order stencil codes using grid sizes of $100^{3}$.
\end{itemize}

Each row in Table~\ref{tab:speedup} quantifies the speedup we achieved by applying the corresponding optimization suggested by GPA.
For each benchmark, we focused on the dominant GPU kernel and implemented one of the top five optimization suggestions, based on its estimated speedup and ease of implementation.
On average, we achieved a geometric mean of 1.22$\times$ speedup with individual speedups ranging from 1.01$\times$ to 3.53$\times$.
GPA's estimated speedup is close to the speedup we achieved, with a geometric mean of the gap between the speedup we achieved and the estimated speedup of 4.1\%.
In the rest of this section, we describe observations while analyzing and optimizing benchmarks using GPA, including the optimization workflow, false positivity, and single dependency coverage.

\subsection{Optimization Workflow}
Before using GPA, one can apply a source-to-source transformation to separate variables that appear on a single line.
Then, one can start by interpreting the top optimizations in the advice report by GPA.
Not all optimizations are easy to implement.
For example, for a code reordering suggestion, if the distance between the def and use instructions is long, it is hard to improve it further.
Based on our experience of studying benchmarks, one can investigate the problem, modify the code, and achieve speedup within half an hour.
Typically, only a few lines need to be changed to achieve non-trivial speedups.

\subsection{False Positivity}
GPA could overestimate optimization opportunities. 
From Table~\ref{tab:speedup}, we observe that loop unrolling and code reordering optimizations have the highest estimate errors.

The overestimation of the benefits of loop unrolling occurs because the loop unrolling optimizer lacks information about the number of iterations and compiler information.
After closely investigating the \textit{bfs} benchmark, we found that the workload is highly unbalanced such that most threads only execute less than four iterations of the loop.
Thus, loop unrolling benefits only a small number of threads.

The data dependency restriction causes the false positivity of code reordering optimization.
GPA suggests reordering a global memory read in a loop of the \textit{pathfinder} benchmark.
The estimated speedup is 26\% higher than we achieved because instructions after synchronizations depend on the results before synchronizations. 
Therefore, the instructions we can use to hide latency are limited in a fine-grained scope in which the distance between the dependent instruction pairs is short no matter how we arrange instructions.

\subsection{Single Dependency Coverage}

\begin{figure*}[htbp]
\centering
\includegraphics[width=0.82\linewidth]{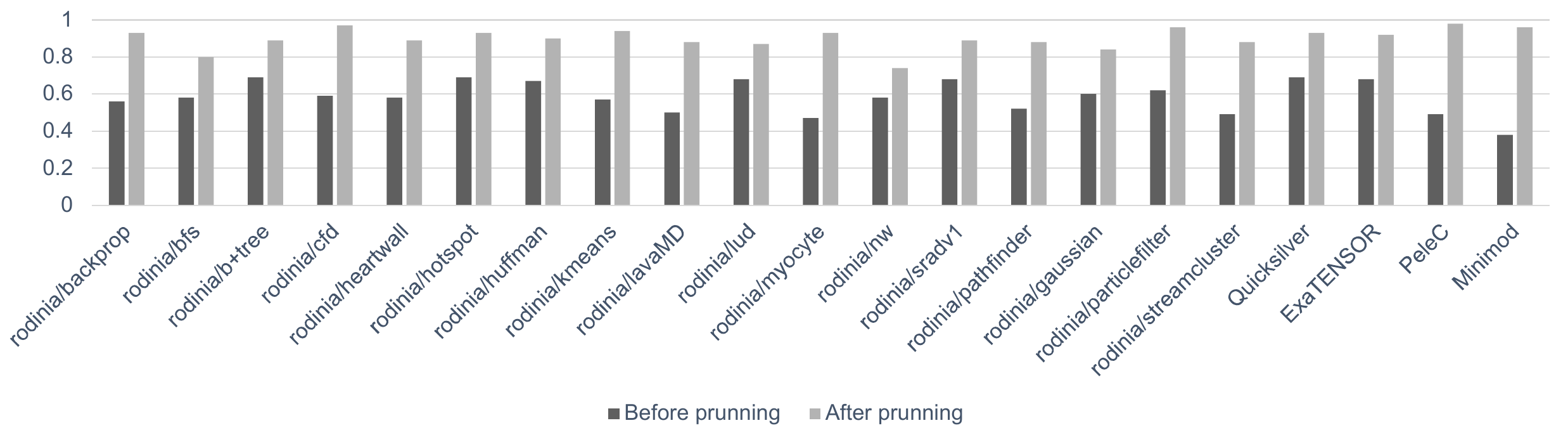}
\caption{Single dependency coverage before and after pruning cold edges}
\label{fig:single dependency}
\end{figure*}

In the instruction dependency graph, we say a node is a \textit{single dependency} node if the node does not have any incoming edge, or each incoming edge represents a different dependency.
We define \textit{single dependency coverage} as the ratio of \textit{single dependency} nodes to the total number of nodes.
Figure~\ref{fig:single dependency} quantifies the single dependency coverage before and after pruning cold edges.
After applying edge pruning heuristics, most benchmarks have single dependency coverage greater than 0.8 so that we can attribute the stalls to one edge without apportioning.

Two exceptions are the \textit{bfs} and the \textit{nw} benchmarks.
The \textit{bfs} benchmark is memory-intensive. Most of the instructions are global memory read/stores, which have a 64-bit memory address stored in two 32-bit registers.
The \textit{nw} benchmark has many nodes with multiple incoming edges because of its intricate control flow.
The dominant loop in \textit{nw} is fully unrolled. Within the loop, there is a condition that decides if values are calculated or not.
If yes, it compares four candidates and chooses the maximum one.

\section{Case Studies}~\label{sec:Case Studies}
In this section, we study the optimizations for the four larger benchmark codes in Table~\ref{tab:speedup},  including ExaTENSOR, Quicksilver, PeleC, and Minimod on the platform we mentioned in Section~\ref{sec:Evaluation}.
The GPU code of the applications was compiled with \texttt{-O3 -lineinfo}.
With the following case studies, we show that one can achieve non-trivial speedup without in-depth knowledge of the assembly code and the GPU architecture.

\subsection{ExaTENSOR}
\begin{figure}[htbp]
\centering
\includegraphics[width=1.05\linewidth]{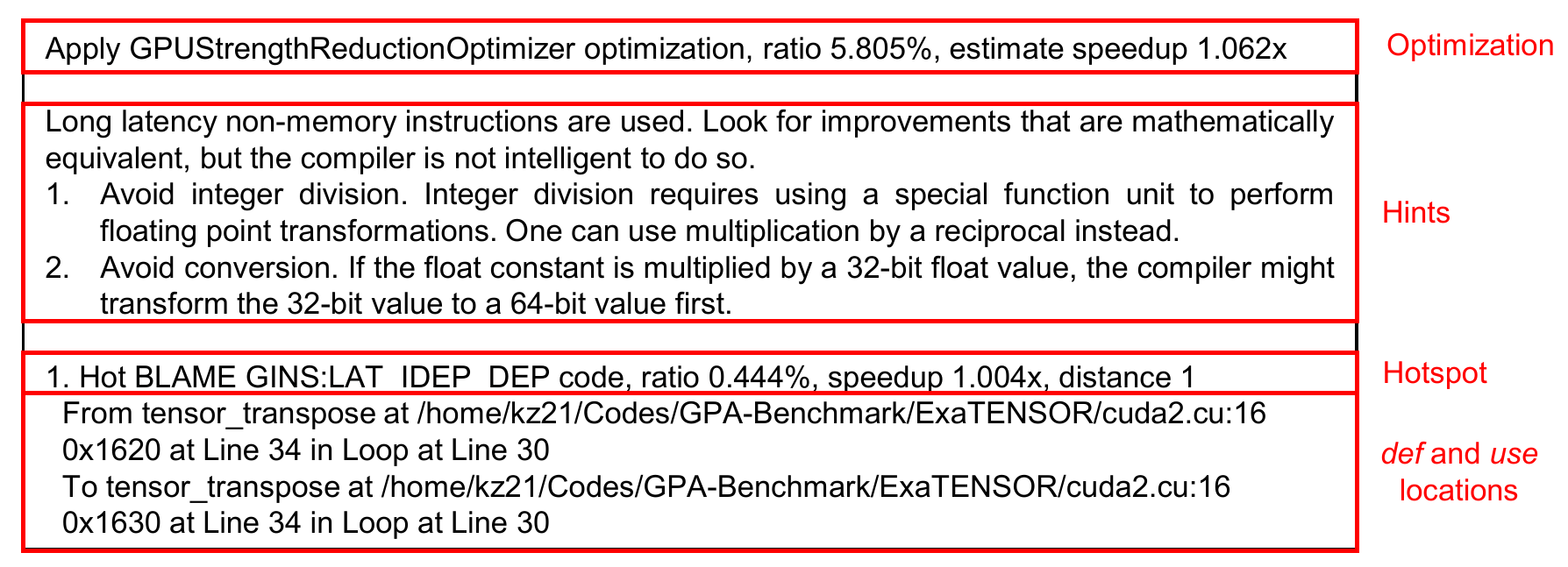}
\caption{A performance report for ExaTENSOR}
\label{fig:gpa report}
\end{figure}

We studied a tensor transpose kernel in ExaTENSOR.
We show a part of GPA's report in Figure~\ref{fig:gpa report}.
GPA ranks optimizers based on their estimated speedups.
Each optimizer suggests a few methods to modify the code and lists several hotspots to focus on.
Each hotspot consists of the \textit{def} and \textit{use} locations and their distance.
In Figure~\ref{fig:gpa report}, GPA reports that we can follow the suggestions of the strength reduction optimizer.
Because the hotspot code performs an integer division, we can replace it with a multiplication by its reciprocal.
This optimization renders a 1.11$\times$ speedup.

We analyzed the modified code again with GPA.
This time GPA suggests a memory transaction reduction optimization to mitigate memory throttling stalls.
In particular, GPA suggests that we replace global memory reads by constant memory reads if elements are shared between threads and not changed during execution.
According to the suggestion, we achieved a 1.03$\times$ speedup.

\subsection{Quicksilver}
We used GPA to analyze Quicksilver on a single GPU.
GPA reports the function inlining optimization may render the highest speedup.
Applying the \texttt{always\_inline} keyword for these functions fails to inline due to the size/register limitation forced by the compiler.
Therefore, we manually inlined two small functions by integrating the whole function bodies into their callers.
By modifying the code in this way, we obtained a $1.14\times$ speedup.

Next, GPA's register reuse optimizer indicates local memory stalls in a loop and points out the potential cause of register spilling.
GPA suggests splitting the loop into two to save registers.
Without GPA, the raw PC sampling report by other tools only show global memory stalls without identifying register pressure.
Applying the optimization yields a $1.01\times$ speedup.

\subsection{PeleC}
We studied the \texttt{react\_state} kernel of PeleC. 
GPA estimates the code reordering optimization may result in the highest speedup.
However, because the top five hotspots only account for 4~\% all of the matched stalls, there are many hotspots distributed across lines so that it is difficult to adjust the code manually.
The second best optimizer suggests increasing the number of blocks.
Since the kernel only occupies 16 blocks, GPA suggests reducing the number of threads per block while increasing the number of blocks to improve the parallelism. 
By increasing the number of blocks to 32, we achieved a 1.21$\times$ speedup.

\subsection{Minimod}
We applied GPA to analyze the \texttt{target\_pml\_3d} kernel of Minimod, which performs higher-order multi-statement stencil computations.
GPA first suggests using the fast math functions to replace high precision match functions.
We applied the \texttt{--use\_fast\_math} compiler flag to achieve a 1.03$\times$ speedup.

Next, GPA suggests the code reordering optimizations for the updated code.
Adjusting the code to read subscripted values from global memory well in advance of their use hides more of the memory latency and yields an additional 1.04$\times$ speedup. 

\section{Related Work}~\label{sec:Related Work}
GPU profilers are widely available in various GPU architectures.
NVIDIA provides several tools~\cite{nvprof, nsightsystem, nsightcompute} to measure GPU performance metrics.
Intel develops VTune~\cite{reinders2005vtune} to monitor executions on both CPUs and GPUs.
AMD provides ROCProfiler~\cite{rocprofiler} to read hardware counters and trace applications.
There are also tools~\cite{hpctoolkit,shende2006tau,scorep,cudablamer-thesis18,cudablamer-protools19} that focus on large HPC applications.
Among the above tools, NVIDIA's nsight-compute provides the most information at the GPU kernel level.
It characterizes GPU kernels' bottlenecks at the high level but does not pinpoint bottlenecks and provide suggestions for specific code regions.
In contrast, GPA analyzes control flow, program structure, architectural features, and interprets measurement data to provide thorough suggestions and estimate speedups.

\sloppy
GPU vendors have also developed instrumentation tools~\cite{kambadur2015fast, stephenson2015flexible, villa2019nvbit, sanitizer} for fine-grained performance measurement and analysis.
These tools, however, introduce unavoidable overhead for GPU kernels.
GPA adopts PC sampling~\cite{cuptipcsampling}, which introduces considerably less cost for kernel execution.
There have been efforts that use instrumentation methods to diagnose specific types of inefficiencies.
Yeh et al. ~\cite{braun2019cuda} instrument GPU code as it is generated by LLVM to identify redundant instructions.
CUDAAdvisor~\cite{shen2018cudaadvisor} also instruments code as it is generated by LLVM to monitor GPU memory access and decide if bypassing could be used.
GVProf~\cite{gvprof} instruments GPU binaries to detect both temporal and spatial redundant value patterns.
These tools only identify a particular type of inefficiencies and do not correlate the problem with hotness.
In comparison, GPA performs a comprehensive analysis of stall reasons for instruction samples and derives various optimization suggestions for hot code regions.

On the CPU side, there exist several tools that examine code quality and provide optimization suggestions.
PerfExpert~\cite{burtscher2010perfexpert} collects performance metrics using sampling, analyzes measurement data and system parameters, and estimates performance upper-bounds.
AutoScope~\cite{sopeju2011autoscope} extends PerfExpert to suggest optimization strategies based on the detected bottlenecks.
Unlike these two tools, CQA~\cite{charif2014cqa} builds a static model by emulating processor pipelines to check symptoms (e.g., vectorization) on the loop level.
VTune~\cite{yasin2014top} uses structured guidance to characterize the bottlenecks by interpreting performance counters. 

Profile-guided optimization takes measurement data as input to guide compiler perform code transformation.
Practical Path Profiling (PPP)~\cite{bond2005practical} collects edge profiles using instrumentation to help compilers make decisions about function inlining and loop unrolling.
Instrumentation-based methods require using representative inputs to dump meaningful profiles.
To avoid the overhead of instrumentation-based approaches, AutoFDO~\cite{chen2016autofdo} uses hardware counter based sampling to collect profiles for production applications and use the profiles to guide optimizations.
While most profile-guided optimization tools attribute measurement data to source lines to provide feedback for compilers, BOLT~\cite{panchenko2019bolt} is a post link optimizer that attributes samples on machine instructions and uses this information to derive binary optimizations.
Recently, there also has been research that incorporates machine learning to guide optimizations.
Cavazos et al.~\cite{cavazos2007rapidly} use profile data as input features to a regression model that predicts the best compiler flags.
DeepFrame~\cite{guha2019deepframe} incorporates deep learning methods to learn the most likely paths during execution and offload the regions to FPGAs.
Though profiler-guided optimizations can automatically adjust code based on rules or models, they only cover a subset of all the available optimizations.
Many optimizations on GPUs need manual effort, such as warp balance, memory coalescing, and adjustments to the thread counts.
Unlike other tools, GPA depends only on line-mapping information and is not tied to any specific compiler.

\section{Conclusions and Future Work}~\label{sec:Conclusion}

Tuning GPU kernels is difficult due to the complexity of GPU architectures and programming models.
To free application developers from needing to interpret measurements from multiple performance counters and analyze program inefficiencies, we introduce GPA.
This performance advisor provides insightful optimization advice at the levels of lines, loops, and kernels and estimates each optimization's speedup.
GPA is organized in a modular fashion.
Users can add custom optimizers to match other inefficiency patterns (e.g., texture fetch combination).

GPA suffers from both hardware and software limitations.
First, GPA apportions stalls to multiple dependency sources with an approximation method based on the instruction counts in the paths.
If the underlying hardware implements ``paired sampling''~\cite{dean1997profileme}, we could collect precisely both the stalled instruction and the instruction that causes the stall.
Second, to obtain a more accurate speedup estimate, comprehensive compiler information such as loop unroll count should be considered.
Last, because PC Sampling with NVIDIA's CUPTI serializes kernel executions, GPA's profiler is unable to measure the effect of concurrent kernel execution.

In the future, we plan to ingest compiler information into GPA to perform a more accurate estimate.
In addition, we can use the insights derived from GPA to guide compilers to apply code transformation for large-scale applications with hundreds of tiny hotspots.

\bibliographystyle{ACM-Reference-Format}


\end{document}